\documentclass[runningheads]{llncs}
\usepackage{amsmath,amssymb,amsfonts,amscd,mathrsfs, color, ascmac, mathtools, bm, authblk}
\usepackage[pdftex]{graphicx}
\usepackage{wrapfig}
\usepackage{makeidx}
\usepackage{fancyhdr}
\usepackage{setspace}
\usepackage{url}

\newcommand{\argmax}{\mathrm{argmax}}

\mathtoolsset{showonlyrefs=true}

\begin{document}
\title{
Sink Location Problems in Dynamic Flow Grid Networks\thanks{Supported by JSPS KAKENHI Grant Numbers 19H04068, 23H03349.}
}
%
%
\author{
Yuya Higashikawa \and
Ayano Nishii \and
Junichi Teruyama \and
Yuki Tokuni
}

\authorrunning{Y.Higashikawa et al.}
%
\institute{
Graduate School of Information Science, University of Hyogo, Japan\\
\email{\{higashikawa,junichi.teruyama\}@gsis.u-hyogo.ac.jp}\\
\email{\{ad23f047,af23v006\}@guh.u-hyogo.ac.jp}
}
\maketitle              
\begin{abstract}
A {\em dynamic flow network} consists of a directed graph, where nodes called {\em sources} represent locations of evacuees, and nodes called {\em sinks} represent locations of evacuation facilities.
Each source and each sink are given {\em supply} representing the number of evacuees and {\em demand} representing the maximum number of acceptable evacuees, respectively.
Each edge is given {\em capacity} and {\em transit time}.
Here, the capacity of an edge bounds the rate at which evacuees can enter the edge per unit time, and the transit time represents the time which evacuees take to travel across the edge.
The {\em evacuation completion time} is the minimum time at which each evacuees can arrive at one of the evacuation facilities.
Given a dynamic flow network without sinks, once sinks are located on some nodes or edges, the evacuation completion time for this sink location is determined.
We then consider the problem of locating sinks to minimize the evacuation completion time, called the {\em sink location problem}.
The problems have been given polynomial-time algorithms only for limited networks such as paths~\cite{higashikawa2015multiple,bhattacharya2017improved,benkoczi2021locating}, cycles~\cite{benkoczi2021locating}, and trees~\cite{mamada2006nlog2n,higashikawa2014minimax,chen2016sink}, but no polynomial-time algorithms are known for more complex network classes.
In this paper, we prove that the 1-sink location problem can be solved in polynomial-time when an input network is a grid with uniform edge capacity and transit time. 

\end{abstract}
\section{Introduction}
In recent years, natural disasters such as earthquakes, tsunamis and hurricanes have been occurring more frequently,
and in Japan, the revision of the Basic Act on Disaster Management~\cite{bousai} has accelerated the development of evacuation facilities, i.e., locating tsunami evacuation towers and setting tsunami evacuation buildings, which had been delayed.
On the other hand, when considering evacuation planning for urgent large-scale disasters such as tsunamis, floods, and nuclear power plant accidents, one of the most important issues to be considered is the delay in evacuation time due to traffic congestion.
In fact, it is known that many people were killed in the Great East Japan Earthquake due to delayed evacuation caused by traffic congestion~\cite{kinjo2011survey}.

To attack this issue, the {\em dynamic flow network model} proposed by Ford and Fulkerson~\cite{ford1958constructing} can be applied.
Since the dynamic flow network model can handle the movement of people or objects over time, it is possible to develop an evacuation plan that quantitatively takes traffic congestion into account.
A {\em dynamic flow network} consists of a directed graph, where nodes called {\em sources} represent locations of evacuees, and nodes called {\em sinks} represent locations of evacuation facilities.
Each source and each sink are given {\em supply} representing the number of evacuees and {\em demand} representing the maximum number of acceptable evacuees, respectively.
Each edge is given {\em capacity} and {\em transit time}.
The capacity of an edge bounds the rate at which evacuees can enter the edge per unit time, and the transit time represents the time which evacuees take to travel across the edge.

One of the most fundamental problems in dynamic flow networks is the {\em quickest transshipment problem}.
The objective of this problem is to compute the minimum time by which each evacuees can arrive at one of sinks, i.e., the {\em evacuation completion time}, and to find the optimal flow of evacuees that achieves the evacuation completion time. 
For the quickest transshipment problem, several strongly polynomial-time algorithms have been developed so far~\cite{kamiyama2019discrete,hoppe2000quickest,schloter2022faster}.

Given a dynamic flow network without sinks, once sinks are located on some nodes or edges, the evacuation completion time for this sink location is determined.
We then consider the problem of locating sinks to minimize the evacuation completion time, called the {\em sink location problem}.
The problems have been given polynomial-time algorithms only for limited networks such as paths~\cite{higashikawa2015multiple,bhattacharya2017improved,benkoczi2021locating}, cycles~\cite{benkoczi2021locating}, and trees~\cite{mamada2006nlog2n,higashikawa2014minimax,chen2016sink}, but no polynomial-time algorithms are known for more complex network classes.
In this paper, we address the sink location problem on grid networks.
A grid network can model an actual road network better than the networks studied so far, in that it consists of a number of cycles.
We present the first polynomial-time algorithm for the 1-sink location problem in dynamic flow grid networks with uniform edge capacity and transit time.

We describe the basic ideas of this study.
In our model, a sink can be located on a node or an edge. When a sink is located at a node, the evacuation completion time for that sink can be computed in polynomial-time using the algorithms~\cite{kamiyama2019discrete,hoppe2000quickest,schloter2022faster} for the quickest transshipment problem. Therefore, if one can compute in polynomial-time the optimal sink location for each edge, i.e., the sink location that minimizes the evacuation completion time over the points on the edge, the optimal sink location for the entire network can be obtained in polynomial-time. 
In the following, we deal with the 1-sink location problem on a particular edge of the input network.

This paper is organized as follows. 
In Section~\ref{sec:preliminaries}, we introduce notations, models and known related results used throughout the paper.
In Section~\ref{sec:grid_network}, we will give a polynomial-time algorithm for the 1-sink location problem on a particular edge.
Combining this algorithm and previous results, we will give a polynomial-time algorithm for the 1-sink location problem on a grid network.
We conclude this paper in Section~\ref{sec:conclusion} with some discussions.

\section{Preliminaries}\label{sec:preliminaries}
\subsection{Models}
Let $\mathbb{R}$ and $\mathbb{R}_+$ denote the sets of real values and non-negative real values.
A dynamic flow network $\mathcal{N}$ is given as a 6-tuple $\mathcal{N} = (G=(V,E),S^+,S^-,w,c,\tau)$, 
where $G=(V,E)$ is a directed graph with node set $V$ and edge set $E$,
$S^+\subseteq{V}$ and $S^-\subseteq V$ are sets of sources and sinks, respectively,
function $w \colon S^+ \cup S^-\rightarrow\mathbb{R}$ represents supply of a source or demand of a sink, 
function $c \colon E\rightarrow\mathbb{R}_+$ is capacity of an edge,
and function $\tau \colon E\rightarrow\mathbb{R}_+$ is transit time of an edge.
As for function $w$, for a source $s \in S^+$, $w(s) (\geq{0})$ represents supply of source $s$, and for a sink $s \in S^-$, $w(s) (< 0)$ represents demand of sink $s$.
We define $w(X)\coloneqq\sum_{s\in{X}}w(s)$ for a subset $X\subseteq{S^+ \cup S^-}$.

In this paper, we consider the problem of locating a sink in a dynamic flow grid network.
Here, let us describe an input network $\mathcal{N}$ in our problem.
First, a graph $G=(V,E)$ is a grid, where $V$ consists of $n = N \times N$ nodes and $E$ consists of bidirected edges in between each adjacent nodes of $V$ (see Fig.~\ref{fig:input}).
Furthermore, we assume that all edge capacities and transit times are uniform, thus $c$ and $\tau$ are constant functions.
In the sink location problem, input network $\mathcal{N}$ has no sinks and all nodes are treated as sources.
Therefore, $\mathcal{N}$ is represented as $(G=(V,E),S^+=V,S^-=\emptyset,w,c,\tau)$.
In the following, we abuse $c$ and $\tau$ as non-negative real constants.

\begin{figure*}[b]
    \begin{tabular}{cc}
      \begin{minipage}[t]{0.35\hsize}
        \centering
        \includegraphics[keepaspectratio, height=3.3cm]{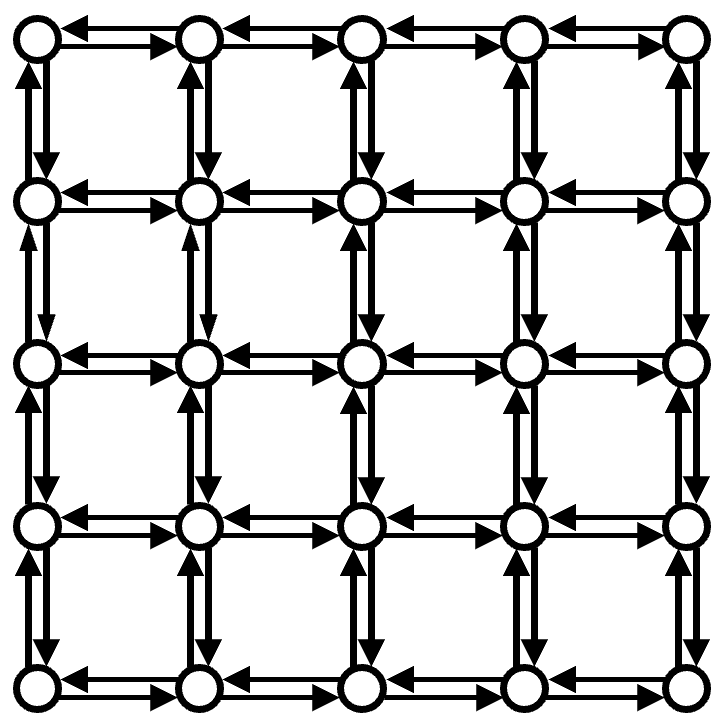}
        \caption{The input network $\mathcal{N}$}
        \label{fig:input}
      \end{minipage} &
      \begin{minipage}[t]{0.6\hsize}
        \centering
        \includegraphics[keepaspectratio, height=3.3cm]{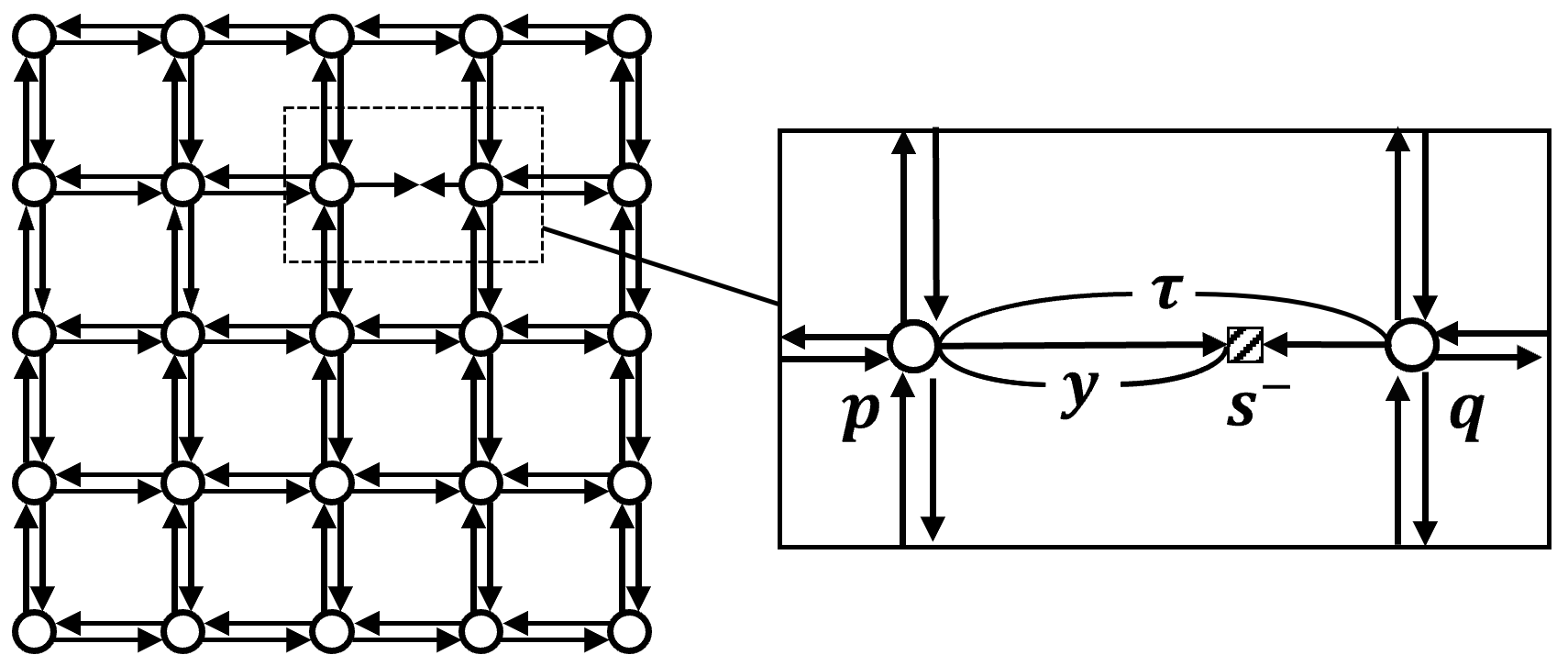}
        \caption{$\mathcal{N}(y)$}
        \label{fig:input_y}
      \end{minipage}
    \end{tabular}
\end{figure*}

Let us consider the problem of locating a sink $s^-$ in between two adjacent nodes $p,q \in V$ at a distance of $y \ (0 < y < \tau)$ from node $p$.
For a given dynamic flow grid network $\mathcal{N}$ without sinks, we apply the following operations to $\mathcal{N}$:
\begin{enumerate}
\item Remove directed edges $(p,q)$ and $(q,p)$.
\item Add a sink $s^-$ with demand $w(s^-) \! = \! - \! \sum_{v\in{V}} \! w(v)$.
\item Add directed edges $(p,s^-)$ with capacity $c$ and transit time $y$, and $(q,s^-)$ with capacity $c$ and transit time $\tau-y$.
\end{enumerate}
The dynamic flow network obtained from the above operations is denoted as $\mathcal{N}(y)$ (see Fig.~\ref{fig:input_y}). 
Hereafter, we denote the sets of nodes and edges in $\mathcal{N}(y)$ as $V(y)$ and $E(y)$, respectively. 
In the following, we consider embedding network $\mathcal{N}(y)$ in a Cartesian coordinate system so that $p$ and $q$ are mapped at $(0,0)$ and $(\tau,0)$, respectively, and the other nodes are mapped at points $(i \tau , j \tau)$ with some integers $i,j$.
For simplicity, a node mapped at $(i \tau , j \tau)$ is called node $(i,j)$.
We refer to the directions from $(0,0)$ to $(0,1)$, from $(0,1)$ to $(0,0)$, from $(0,0)$ to $(1,0)$, and from $(1,0)$ to $(0,0)$ as {\em upward}, {\em downward}, {\em rightward}, and {\em leftward}, respectively.   

\subsection{Evacuation completion time}
A {\em dynamic flow} $f$ is defined as a function $f \colon E(y) \times \mathbb{R}_+ \rightarrow \mathbb{R}_+$, where $f(e,\theta)$ represents the flow rate entering edge $e \in E(y)$ at time $\theta \in \mathbb{R}_+$.
In this paper, we deal with dynamic flows in continuous-time model~\cite{fleischer1998efficient,Higashikawa2019survey}.
Let us consider the following constraints for a dynamic flow $f$:
\begin{equation}
\begin{split}
    0 \leq f(e, \theta) \leq c
    \quad \text{for any} \ e \in E(y), \ \text{and for any} \ \theta \in \mathbb{R}_+,
    \label{eq:capacity_constraint}
\end{split}
\end{equation}
\begin{equation}
\begin{split}
    \sum_{(v,u)\in E(y)} \int_0^\theta f((v,u),t) dt - \sum_{(u,v)\in E(y)} \int_0^{\theta-\tau} f((u,v),t) dt \leq w(v)&\\
    \text{for any} \ v \in S^+ (=V(y) \setminus \{s^-\}), \ \text{and for any} \ \theta \in \mathbb{R}_+. \hspace{-20pt}&
    \label{eq:flow_conservation}
\end{split}
\end{equation}
The constraints \eqref{eq:capacity_constraint} and \eqref{eq:flow_conservation} are called the capacity constraint and the conserve constraint.  
The conserve constraint \eqref{eq:flow_conservation} means that for any time $\theta$ and any source $v$, the amount of flow out of $v$ within time $\theta$ is at most the amount of flow entering $v$ within $\theta$ plus the amount of the supply at $v$. 
Furthermore, for a {\em time horizon} $T \in \mathbb{R}_+$, consider the following constraint:
\begin{equation}
\begin{split}
    \sum_{(v,u)\in E(y)} \int_0^T f((v,u),t) dt - \sum_{(u,v)\in E(y)} \int_0^{T-\tau} f((u,v),t) dt &= w(v)\\
    \text{for any} \ v \in S^+ (=V(y) \setminus \{s^-\}),  \hspace{-55pt}&\\
    \int_0^{T-y} f((p,s^-),t) dt + \int_0^{T-(\tau-y)} f((q,s^-),t) dt &= -w(s^-).
    \label{eq:demand_constraint}
\end{split}
\end{equation}
The constraint \eqref{eq:demand_constraint} implies that for each node $v$, the net amount of flow accumulated at $v$ within time $T$ equals its supply or demand.
For a time horizon $T$, if a dynamic flow $f$ on $\mathcal{N}(y)$ satisfies the above constraints \eqref{eq:capacity_constraint}, \eqref{eq:flow_conservation}, and \eqref{eq:demand_constraint}, 
then $f$ is said to be {\em feasible} for $T$.
Moreover, a time horizon $T$ for which a dynamic flow is feasible is said to be feasible.
%
The evacuation completion time denotes the minimum time for which a feasible dynamic flow exists. 
Letting $\Theta^*(y)$ denote the evacuation completion time in $\mathcal{N}(y)$, the 1-sink location problem between $p$ and $q$ is formulated as follows:
\begin{alignat*}{2}
    \quad & \text{minimize}   \quad && \Theta^*(y)\\
          & \text{subject to} \quad && 0 < y < \tau,
\end{alignat*}
which we call {\sc Sink-Location-on-Edge} (for short, {\sf SLE}).
Next, we describe known properties of the evacuation completion time.
Given a dynamic flow network $\mathcal{N}(y)$, for a subset of the sources $X \subseteq S^+$ and a time horizon $\theta \in \mathbb{R}_+$, let $o_\theta(X,y)$ be the maximum amount of dynamic flow that can reach the sink $s^-$ from $X$ within time horizon $\theta$ (more formal definition in our model will be given in \eqref{eq:o_theta_2}).
The following theorem by Hoppe and Tardos~\cite{hoppe2000quickest} says a property of a feasible time horizon.
%
%
\begin{theorem}[\cite{hoppe2000quickest}]
\label{thm:Klinz}
Given a dynamic flow network $\mathcal{N}(y)$ for a real value $y$ with $0 < y < \tau$, and a time horizon $\theta \in \mathbb{R}_+$, 
there exists a feasible dynamic flow for $\theta$ on $\mathcal{N}(y)$ if and only if
\begin{equation}
    \min\{o_\theta(X,y)-w(X) \mid X \subseteq{S^+}\} \geq 0
    \label{eq:Klinz}
\end{equation}
holds.
%
\end{theorem}
Therefore, the evacuation completion time $\Theta^*(y)$ in $\mathcal{N}(y)$ is the minimum $\theta$ that satisfies \eqref{eq:Klinz}.
Here, we define
\begin{equation}
    \Theta(X,y) \coloneqq \min\{\theta \mid o_\theta(X,y)-w(X)\geq{0}\}
    \label{eq:Theta_(X,y)}
\end{equation}
and then, because $o_\theta(X,y)-w(X)$ is monotonically non-decreasing in $\theta \in \mathbb{R}_+$ for any $X \subseteq S^+$ and any $y$, we have
\begin{equation}
    \Theta^*(y)=\max\{\Theta(X,y) \mid {X}\subseteq{S^+}\}.
    \label{eq:Theta^*(y)}
\end{equation}

\subsection{Residual networks}
The dynamic flow network $\mathcal{N}(y)$ can be treated as a static flow network $\overline{\mathcal{N}}(y)$ by considering the transit time as the cost of each edge.
We say that a static flow $\bar{f} \colon E(y) \rightarrow \mathbb{R}_+$ is feasible with respect to $X \subseteq S^+$ if $\bar{f}$ satisfies the following conditions:
\begin{align*}
    0 \leq \bar{f}(e) &\leq c
    && \text{for any} \ e \in E(y),\\
    \sum_{(v,u) \in E(y)} \bar{f}(v,u) - \sum_{(u,v) \in E(y)} \bar{f}(u,v) &= 0 
    && \text{for any} \ v \in V(y) \setminus (X \cup \{s^-\}),\\
    \sum_{v\in X} \sum_{(v,u) \in E(y)} \bar{f}(v,u) - 
    \sum_{u \in \{p,q\}} \bar{f}(u,s^-) &= 0. 
    &&\quad  
\end{align*}
When a feasible static flow $\bar{f} \colon E(y) \rightarrow \mathbb{R}_+$ is given to $\overline{\mathcal{N}}(y)$, the {\em residual network} with respect to $\bar{f}$, denoted by  $\overline{\mathcal{N}}(y)_{\bar{f}}$, is constructed as follows.
For each edge $e$ with $\bar{f}(e) > 0$, a reverse edge $\overset{\leftarrow}{e}$ is added, and the edge set obtained by adding such reverse edges to $E(y)$ is denoted by $E(y)_{\bar{f}}$.
The capacity of residual edges $c_{\bar{f}} \colon E(y)_{\bar{f}}\rightarrow{\mathbb{R}_+}$ is defined so that $c_{\bar{f}}(e)=c-{\bar{f}}(e)$ for original edges $e \in E(y)$, and $c_{\bar{f}}(\overset{\leftarrow}{e})={\bar{f}}(e)$ for reverse edges $\overset{\leftarrow}{e} \in E(y)_{\bar{f}} \setminus E(y)$.
The cost of residual edges $\tau_{\bar{f}} \colon E(y)_{\bar{f}}\rightarrow{\mathbb{R}}$ is defined as follows:
For edges $(p,s^-)$, $(s^-,p)$, $(q,s^-)$, and $(s^-,q)$, $\tau_{\bar{f}}(p,s^-) = y$, $\tau_{\bar{f}}(s^-,p) = - y$, $\tau_{\bar{f}}(q,s^-) = \tau - y$, and $\tau_{\bar{f}}(s^-,q)= - \tau + y$ respectively.
For other original edges $e \in E(y)$, $\tau_{\bar{f}}(e) = \tau$, and for reverse edges $\overset{\leftarrow}{e} \in E(y)_{\bar{f}} \setminus E(y)$, $\tau_{\bar{f}}(\overset{\leftarrow}{e}) = -\tau$.


\subsection{Envelope of two-dimensional line segments}
Let $\mathcal{F}$ be a family of one-variable linear functions $f_1, \ldots, f_n$ defined on closed intervals $[a_1,b_1], \ldots,$ $[a_n,b_n]$, respectively.
We simply refer to such $\mathcal{F}$ as a set of two-dimensional line segments.
Moreover, for each function $f\in\mathcal{F}$ defined on a closed interval $[a,b]$, we define $\underline{f}(x)$ as follows:
\begin{equation*}
    \underline{f}(x) \coloneqq \left\{
    \begin{array}{ll}
    \! f(x) & (x \in [a,b]), \\
    \! -\infty & (x \notin [a,b]).
    \end{array}
    \right.
\end{equation*}
The {\em upper envelope} $U_{\mathcal{F}}$ of the two-dimensional line segments $\mathcal{F}$ is defined as
\begin{equation*}
    U_{\mathcal{F}}(x)  \coloneqq \max_{1\leq{i}\leq{n}} \underline{f_i}(x) 
\end{equation*}
that is defined on $\bigcup_{1 \le i \le n} [a_i,b_i]$.
A point $(x',U_{\mathcal{F}}(x'))$ in the upper envelope $U_{\mathcal{F}}$ is called a {\em break point} if a function constituting $U_{\mathcal{F}}$ switches at $x'$, that is, if
\begin{equation*}
    \lim_{x \to x'+0} \argmax_{1\leq{i}\leq{n}} \underline{f_i}(x) 
    \neq \lim_{x \to x'-0} \argmax_{1\leq{i}\leq{n}} \underline{f_i}(x).
\end{equation*}
It is known that the following theorem holds for the upper envelope $U_{\mathcal{F}}$ of the set of two-dimensional line segments $\mathcal{F}$.
\begin{theorem} 
[\cite{hart1986nonlinearity,hershberger1989finding}]
\label{thm:upperenvelope}
Given a set $\mathcal{F}$ containing $n$ two-dimensional line segments, the upper envelope $U_{\mathcal{F}}$ has at most $O(n\alpha(n))$ break points, and they can be computed in $O(n\log n)$ time.
Here, the function $\alpha$ denotes the inverse function of the Ackermann function.
\end{theorem}
In the upper envelope $U_{\mathcal{F}}$ of a two-dimensional line segments $\mathcal{F}$, a line segment connects each pair of adjacent break points.
Then, this implies that some of break points gives the minimum value of $U_{\mathcal{F}}$, and we have the following corollary.
\begin{corollary}
\label{cor:upperenvelope}
Given a two-dimensional line segments $\mathcal{F}$ containing $n$ line segments, the break point that minimizes the upper envelope $U_{\mathcal{F}}$ can be computed in $O(n\log n)$ time.
\end{corollary}

\section{Sink location on an edge}\label{sec:grid_network}
In this section, we propose a polynomial-time algorithm for problem {\sf SLE}.
According to \eqref{eq:Theta^*(y)}, if we can compute the function $\Theta(X,y)$ in $y$, for all $X \subseteq S^+$, then the objective function $\Theta^*(y)$ can be represented as the upper envelope of $\Theta(X,y)$ over $0 < y < \tau$.
However, this brute-force method does not provide a polynomial-time algorithm since there are exponentially many subsets $X$.
On the contrary, we give a family of $O(\sqrt{n})$ source sets that must contain $X$ maximizing $\Theta(X,y)$ for any $y$ with $0 < y < \tau$, and based on this property, provide a polynomial-time algorithm.
In Section \ref{sec:function_property}, we show properties of $\Theta(X,y)$, in Section \ref{sec:domination_property}, we show properties of $X$ that maximizes $\Theta(X,y)$, and in Section \ref{sec:algorithm}, a polynomial-time algorithm is provided.



\subsection{Properties of $\Theta(X,y)$} 
\label{sec:function_property}
First, we introduce the notation used in the following.
For any two sources $u,v \in S^+$, we define $d(u,v)$ as the length of the shortest path from $u$ to $v$ in the static flow network $\overline{\mathcal{N}}(y)$. 
Note that the shortest path from $u$ to $v$ does not use edges $(p,s^-)$ and $(q,s^-)$ in $E(y)$. For any source set $X \subseteq S^+$ and source $v \in S^+$, we define $d(X,v)$ as the length of the shortest path from $X$ to $v$ in $\overline{\mathcal{N}}(y)$, that is,
\begin{equation}
    d(X,v) \coloneqq \min\{d(x,v) \mid x\in{X}\}.
    \label{eq:d_uv}
\end{equation}

Here, we discuss the properties of $o_\theta(X,y)$ which represents the maximum amount of dynamic flow that can reach the sink $s^-$ from $X$ within time horizon $\theta$. 
It is known that $o_\theta(X,y)$ is characterized by a minimum-cost flow from $X$ to $S^-$ in the static flow network $\overline{\mathcal{N}}(y)$~\cite{ford1958constructing}.
Specifically, given $X \subseteq S^+ \cup S^-$,
$o_\theta(X,y)$  as a function in $\theta$ can be computed by applying the successively shortest path algorithm ~\cite{jewell1958optimal,busacker1960procedure,iri1960new} for the minimum-cost flow problem in the following manner.
Initially, set $\bar{f}$ as a zero static flow, i.e., $\bar{f} : E(y) \rightarrow 0$. 
At each step $i \ (\geq 1)$, execute the following two procedures.
\begin{enumerate}
    \item Find the shortest (i.e., minimum-cost) path $P_i(X,y)$ from $X$ to $s^-$ in the current residual network $\overline{\mathcal{N}}(y)_{\bar{f}}$. If there is no such path, then break the iteration.
    \item Update $\bar{f}$ by adding a static flow of amount $c$ along path $P_i(X,y)$.
\end{enumerate}
%
Note that the number of paths obtained when the above iteration halts is exactly two.
Letting $|P_1(X,y)|$ and $|P_2(X,y)|$ be the lengths of paths $P_1(X,y)$ and $P_2(X,y)$, respectively, $o_\theta(X,y)$ is represented by the following equation.
\begin{equation}
\begin{split}
    o_{\theta}(X,y) = 
    \left\{
    \begin{array}{ll}
    0
    & 0 \leq \theta < |P_1(X,y)|,\\
    c(\theta  -  |P_1(X,y)|)
    & |P_1(X,y)| \leq \theta < |P_2(X,y)|,\\
    c(\theta - |P_1(X,y)|) + c(\theta - |P_2(X,y)|) \:\:
    &|P_2(X,y)| \leq \theta.
    \end{array}
    \right.
    \label{eq:o_theta_2}
\end{split}
\end{equation}
Since \eqref{eq:o_theta_2} implies that $o_\theta(X,y)$ is a piecewise linear convex function in $\theta$, it can be transformed into the following:
\begin{equation}
\begin{split}
    o_\theta(X,y) =  
    \max \left\{ 0, c(\theta \!-\!|P_1(X,y)|), 
    c(\theta \!-\! |P_1(X,y)|)+c(\theta \!-\! |P_2(X,y)|)
    \right\}.
    \label{eq:o_theta_0}
\end{split}
\end{equation}

Note that $P_1(X,y)$ clearly does not contain any reverse edge 
since $P_1(X,y)$ is a path in $\overline{\mathcal{N}}(y)$.
In fact, not only that, there exists $P_2(X,y)$ not containing any reverse edge, which follows Lemma~\ref{path} below.
Thus, it is not necessary to consider the residual network when calculating $o_\theta(X,y)$.
Here, we define $d_1(X,y)$ as the minimum of the shortest path lengths that pass through $p$ or $q$, respectively, from $X$ to $s^-$ on $\overline{\mathcal{N}}(y)$. 
We define $d_2(X,y)$ as the maximum of these shortest path lengths. 
We thus have
\begin{equation}
\begin{split}
    d_1(X,y) &\coloneqq \min\{d(X,p)+y,d(X,q)+\tau-y\}, \text{ and}\\
    d_2(X,y) &\coloneqq \max\{d(X,p)+y,d(X,q)+\tau-y\}.
    \label{eq:d_1,d_2}
\end{split}
\end{equation}

As for paths in $\overline{\mathcal{N}}(y)$ corresponding to $d_1(X,y)$ and $d_2(X,y)$,
we present the following lemma.

\begin{lemma}\label{path}
For any source set $X\subseteq{S^+}$, there exists on $\overline{\mathcal{N}}(y)$ a pair of edge-disjoint paths from $X$ to $y$ whose lengths are $d_1(X,y)$ and $d_2(X,y)$, respectively.
\end{lemma}

\begin{proof}
\begin{figure*}[b]
    \centering
    \includegraphics[width=0.7\hsize, keepaspectratio]{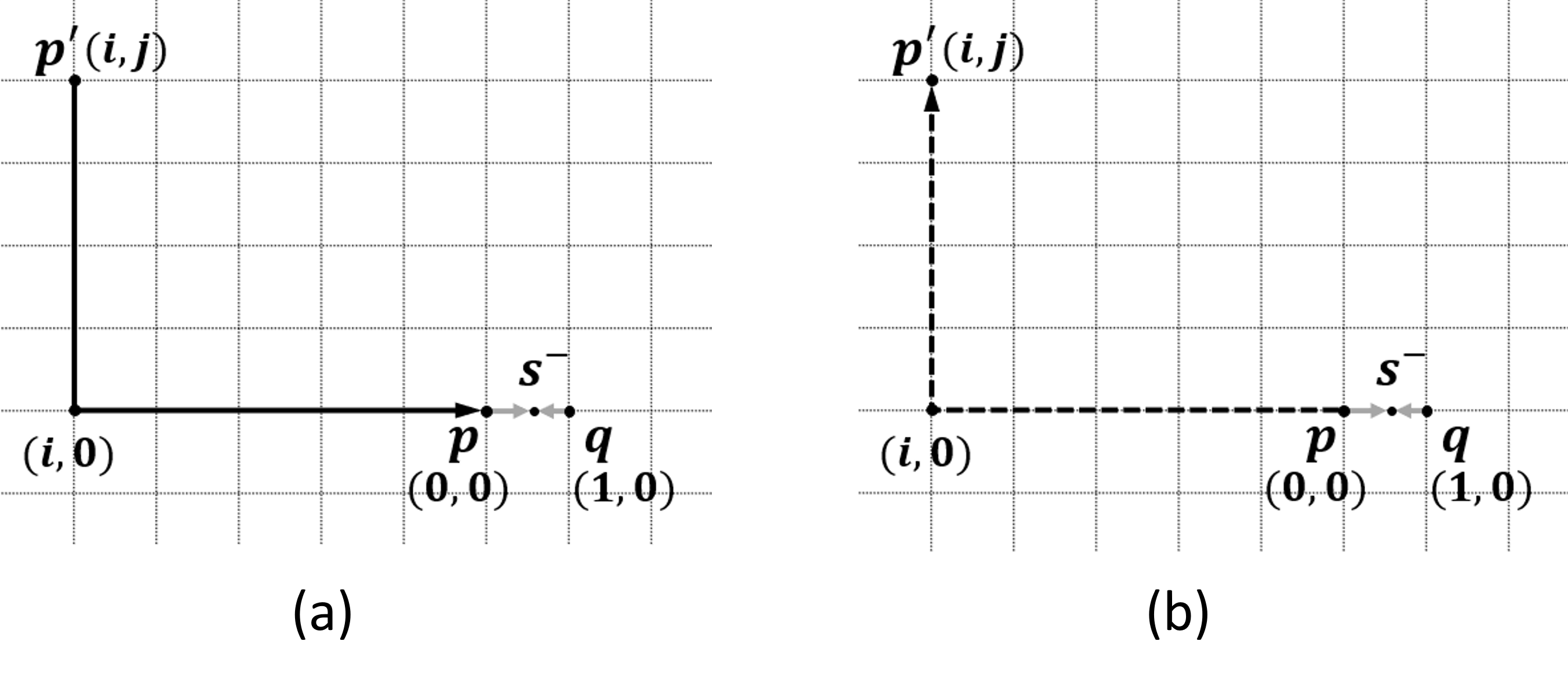}
    \caption{(a) The shortest path in $\mathcal{N}(y)$ from $p'$ to $p$ that passes through $(i,0)$. \\
    (b) Positions where reverse edges exist in the residual network $\overline{\mathcal{N}}(y)_{\bar{f}}$.}
    \label{fig:path_u}
\end{figure*}

First, for any pair of paths corresponding to $d_1(X,y)$ and $d_2(X,y)$, if one path passes through $p$, then the other path passes through $q$. 
In other words, if one path includes the directed edge $(p,s^-)$, then the other path includes the directed edge $(q,s^-)$.
Therefore, from equation \eqref{eq:d_1,d_2}, it is sufficient to show the existence of a pair of shortest paths from $X$ to $p$ and $X$ to $q$ in $\overline{\mathcal{N}}(y)$ that are edge-disjoint.
In the following discussion, we only consider the case where $d(X,p) \leq d(X,q)$. 
Note that a similar argument holds even when $d(X,p) > d(X,q)$.

Let $p'=(i,j)\in X$ be the starting node of the shortest path from $X$ to $p$.
Note that due to the assumption $d(X,p) \leq d(X,q)$, we have $i\leq 0$.
Among shortest paths from $X$ to $p$, we choose the unique path passing through node $(i,0)$, and let $\bar{f}$ be the static flow of amount $c$ along this path (see Fig.~\ref{fig:path_u}(a)).

In the residual network $\overline{\mathcal{N}}(y)_{\bar{f}}$ (see Fig.~\ref{fig:path_u}(b)), let $q'=(i',j') \in X$ be the starting node of the shortest path from $X$ to $q$.
Let $P^*_q$ be the shortest path in $\overline{\mathcal{N}}(y)_{\bar{f}}$ from $q'=(i',j')$ to $q=(1,0)$ via $(1,j')$ that is guaranteed not to contain any reverse edge, we thus have
\begin{equation}
    |P^*_q|=(|i'-1|+|j'|)\tau.
    \label{eq:|P^*_q|}
\end{equation}
In the following, we show that $P^*_q$ is one of shortest paths in $\overline{\mathcal{N}}(y)_{\bar{f}}$ from $q'$ to $q$, i.e., the length of any other path is at least $(|i'-1|+|j'|)\tau$.

Let $P_q$ be any path from $q'$ to $q$ in $\overline{\mathcal{N}}(y)_{\bar{f}}$.
Let us consider vertical and horizontal lines passing through node $q$.
We then call the half plane to the left (resp. right) of the vertical line {\em L-region} (resp. {\em R-region}),
and similarly the half plane to above (resp. below) the horizontal line {\em U-region} (resp. {\em D-region}).
Here notice that $P_q$ consists of four types of edges: leftward, rightward, upward and downward edges.
Let $r_1$ (resp. $l_1$) be the number of rightward (resp. leftward) edges of $P_q$ lying in R-region.
Let $r_2$ (resp. $l_2$) be the number of rightward (resp. leftward) edges of $P_q$ lying in L-region.
Let $u_1$ (resp. $d_1$) be the number of upward (resp. downward) edges of $P_q$ lying in U-region.
Let $u_2$ (resp. $d_2$) be the number of upward (resp. downward) edges of $P_q$ lying in D-region.
In the residual network $\overline{\mathcal{N}}(y)_{\bar{f}}$, there are leftward reverse edges from $p=(0,0)$ to $(i,0)$ in L-region
and upward reverse edges from $(i,0)$ to $p'=(i,j)$ in U-region.
Therefore, the number of reverse edges in $P_q$ is at most $u_1 + l_2$, we thus have
\begin{equation}
    |P_q| \geq (l_1-l_2+r_1+r_2-u_1+u_2+d_1+d_2)\tau.
    \label{eq:P_q-1}
\end{equation}
First, we show that
\begin{equation}
l_1 - l_2 + r_1 + r_2 = 2 r_1 + |i'-1|.
\label{eq:right_left}
\end{equation}
When $i' \leq 1$, we have
\begin{equation}
l_1 = r_1, \quad l_2 = r_2-|i'-1|,
\label{eq:i_1}
\end{equation}
and when $i' > 1$, we have
\begin{equation}
\begin{split}
l_1 = r_1+|i'-1|, \quad l_2 = r_2.
\label{eq:i_2}
\end{split}
\end{equation}
In both cases of~\eqref{eq:i_1} and \eqref{eq:i_2}, equation \eqref{eq:right_left} holds.
Similarly, we show that 
\begin{equation}
- u_1 + u_2 + d_1 + d_2 = 2 d_2 + |j'|.
\label{eq:up_down}
\end{equation}
When $j' \leq 0$, we have
\begin{equation}
u_1 = d_1, \quad u_2 = d_2+|j'|,
\label{eq:j_1}
\end{equation}
and when $j' > 0$, we have
\begin{equation}
u_1 = d_1-|j'|, \quad u_2 = d_2.
\label{eq:j_2}
\end{equation}
In both cases of~\eqref{eq:j_1} and \eqref{eq:j_2}, equation \eqref{eq:up_down} holds.
Therefore, by equations \eqref{eq:P_q-1}, \eqref{eq:right_left} and \eqref{eq:up_down}, we have
\begin{equation}
\begin{split}
    |P_q| &\geq (2r_1+2d_2+|i'-1|+|j'|)\tau\\
          &\geq (|i'-1|+|j'|)\tau.
    \label{eq:|P_q|}
\end{split}
\end{equation}
The last inequality is led by the fact that $r_1$ and $d_2$ are non-negative.
Therefore, by equations \eqref{eq:|P_q|} and \eqref{eq:|P^*_q|}, we can conclude that $P^*_q$ is the shortest path in $\overline{\mathcal{N}}(y)_{\bar{f}}$ from $q'$ to $q$.
\qed
\end{proof}


%

As mentioned above, by Lemma~\ref{path}, we have $|P_1(X,y)| = d_1(X, y)$ and $|P_2(X,y)|$ $= d_2(X,y)$.
Substituting these into~\eqref{eq:o_theta_0}, $o_\theta(X,y)$ is represented by
\begin{equation}
\begin{split}
    \hspace{-10pt}
    o_\theta(X,y) = \max \left\{ 
    0,
    c \left(\theta - d_1(X,y) \right),
    c \left(\theta - d_1(X,y) \right) + c \left( \theta - d_2(X,y) \right)
    \right\}.
    \label{eq:max_o_theta}
\end{split}
\end{equation}
By the definition of $d_1(X,y)$, we obtain
\begin{equation}
\begin{split}
    c \left(\theta - d_1(X,y) \right)
    &= c \left(\theta - \min \{ d(X,p)+y, d(X,q)+\tau-y\} \right)\\
    &= \max \left\{ c(\theta - d(X,p) - \! y), c(\theta - d(X,q) - \tau + y ) \right\}.
    \label{eq:c_theta}
\end{split}
\end{equation}
Since 
\[
d_1(X,y)+d_2(X,y)=d(X,p)+d(X,q)+\tau
\]
holds by \eqref{eq:d_1,d_2}, we have
\begin{equation}
\begin{split}
    c \left(\theta - d_1(X,y) \right) + c \left( \theta - d_2(X,y) \right)
    = c \left(2\theta - d(X,p) - d(X,q) - \tau \right).
    \label{eq:2c_theta}
\end{split}
\end{equation}
Substituting \eqref{eq:c_theta} and \eqref{eq:2c_theta} into \eqref{eq:max_o_theta}, we obtain
\begin{equation}
\begin{split}
    o_\theta(X,y) =  \max \bigl\{
    & 0, c(\theta - d(X,p)-y),
    c(\theta - d(X,q)-\tau+y), \\
    & c(2\theta-d(X,p)-d(X,q)-\tau)
    \bigr\}.
    \label{eq:o_theta}
\end{split}
\end{equation}
Here, we describe the properties of $\Theta(X,y)$.
According to its definition \eqref{eq:Theta_(X,y)}, $\Theta(X,y)$ is the value of $\theta$ that satisfies $o_{\theta}(X,y)=w(X)$.
Regarding $\Theta(X,y)$ as a function in $y$, we have the following theorem
(see Fig.~\ref{fig:Theta_(X,y)} for the shape of $\Theta(X,y)$ for $0 < y < \tau$).

\begin{figure*}[tb]
    \centering
    \includegraphics[width=\hsize, keepaspectratio]{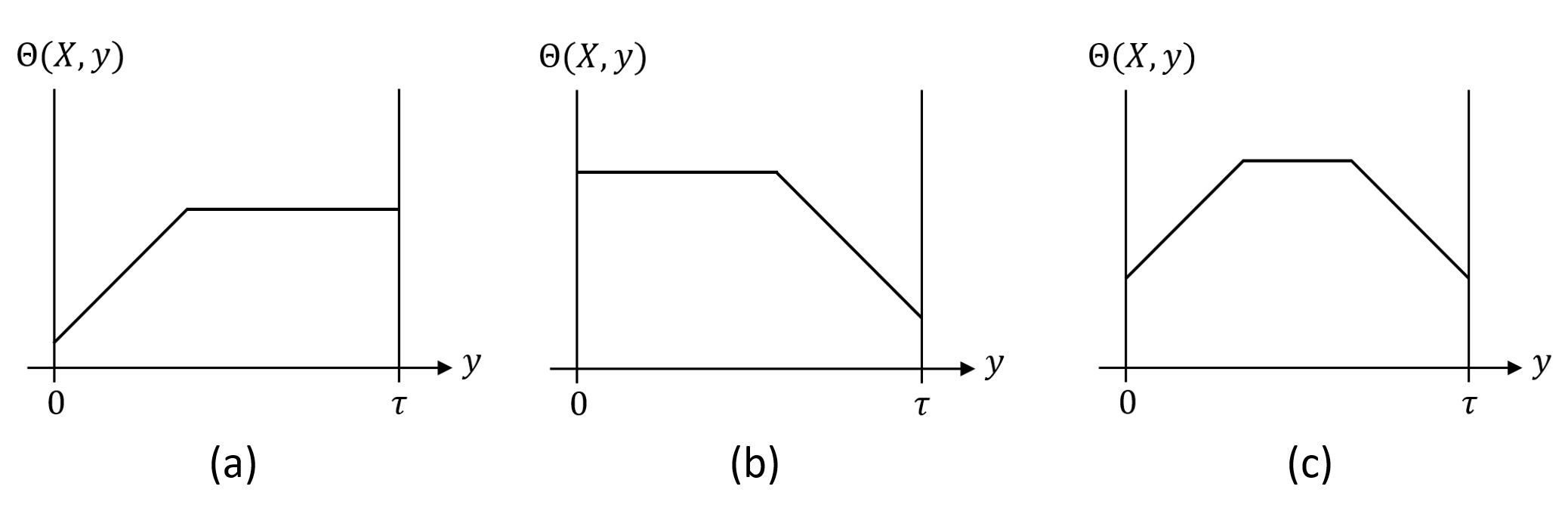}
    \caption{The shape of $\Theta(X,y)$ for $0 < y < \tau$ in the cases where (a) $d(X,p) < d(X,q)$, (b) $d(X,p) > d(X,q)$, and (c) $d(X,p) = d(X,q)$, respectively.}
    \label{fig:Theta_(X,y)}
\end{figure*}

\begin{theorem}\label{thm:Theta(X,y)}
For any source set $X\subseteq{S^+}$, $\Theta(X,y)$ is a piecewise linear function in $y$ with $0<y<\tau$ of at most three line segments, represented by
\begin{equation*}
    \begin{split}
        \Theta(X,y) = \min
        \biggl\{
        & y + \frac{w(X)}{c} + d(X,p),
         -y + \tau + \frac{w(X)}{c} + d(X,q), \\
        & \frac{w(X)}{2c} + \frac{d(X,p) + d(X,q) + \tau}{2}
        \biggr\}.
    \end{split}
\end{equation*}
\end{theorem}

\begin{proof}
Let us denote the linear functions in $\theta$ that compose $o_\theta(X,y)$ as
\begin{equation*}
\begin{split}
    g_1(\theta) &= c\left(\theta - d(X,p)-y\right),\\
    g_2(\theta) &= c\left(\theta - d(X,q)-\tau+y\right),\\
    g_3(\theta) &= c\left(2\theta-d(X,p)-d(X,q)-\tau\right).
\end{split}
\end{equation*}
By equations \eqref{eq:Theta_(X,y)} and \eqref{eq:o_theta}, $\Theta(X,y)$ is the minimum value of $\theta$ among values that make each of $g_1(\theta)$, $g_2(\theta)$, and $g_3(\theta)$ equal to $w(X)$.
By solving $g_1(\theta) = w(X)$ for $\theta$, we obtain $\theta = y + \frac{w(X)}{c} + d(X,p)$.
Similarly, $g_2(\theta) = w(X)$ when $\theta = -y + \tau + \frac{w(X)}{c} + d(X,q)$, and $g_3(\theta) = w(X)$ when $\theta= \frac{w(X)}{2c} + \frac{d(X,p) + d(X,q)+\tau}{2}$.
Therefore, the minimum value among these three values is $\Theta(X,y)$. 
Additionally, with respect to $y$, each of the above $\theta$ is linear function.
Hence, $\Theta(X,y)$ is a piecewise linear function in $y$ consisting of at most three segments.
\qed
\end{proof}

\subsection{Dominant source sets}
\label{sec:domination_property}
For two subsets of $S^+$, $X$ and $X'$, $X$ is said to {\em dominate} $X'$ if $\Theta(X,y) \geq \Theta(X',y)$ holds for all $y$.
Let us consider the sufficient condition for a source set $X$ to dominate another source set $X'$.
If $d(X,p) = d(X',p)$ and $d(X,q) = d(X',q)$, then functions $o_\theta(X,y)$ and $o_\theta(X',y)$ coincide according to \eqref{eq:o_theta}.
In this case, if $w(X) \geq w(X')$, then $\Theta(X,y) \geq \Theta(X',y)$ by definition \eqref{eq:Theta_(X,y)}.
Among source sets $X$ such that $d(X,p)$ and $d(X,q)$ remain the same respectively, one with the maximum $w(X)$ is called a {\em dominant source set}.
For two integers $i$ and $j$, let $X_{i,j}$ be the dominant source set such that $d(X_{i,j},p) = i \tau$ and $d(X_{i,j},q) = j \tau$.
Then, $X_{i,j}$ is formally defined by the following:
\begin{equation*}
    X_{i,j} \coloneqq 
    \{ x\in V(y) \mid d(x,p) \geq i\tau, \ d(x,q) \geq j\tau \}.
\end{equation*}
Note that $X_{i,j}=\emptyset$ may hold when $i$ or/and $j$ are large enough.
Let $\mathcal{X}$ be the family of dominant source sets.
Note that both $d(X,p)$ and $d(X,q)$ are at most $2\tau\sqrt{n}$, we thus have
\begin{equation}
    \mathcal{X} \coloneqq \{X_{i,j} \mid 0 \leq i \leq 2\sqrt{n}, 0 \leq j \leq 2\sqrt{n}\} \setminus \{\emptyset \},
    \label{eq:family}
\end{equation}
and also
\begin{equation}
    \Theta^*(y) = \max\{\Theta(X_{i,j},y) \mid X_{i,j} \in \mathcal{X}\}.
    \label{eq:new_Theta^*(y)}
\end{equation}
Although the number of source sets in $\mathcal{X}$ seems to be $O(n)$ by definition~\eqref{eq:family},
we can provide a more precise upper bound on the number of $X_{i,j}$.
The following lemma implies that when considering $X_{i,j}$ for a fixed $i$, only the constant number of $j$ are enough to be considered.

\begin{lemma}\label{combination}
For any source set $X\subseteq{S^+}$, it holds in $\overline{\mathcal{N}}(y)$
\begin{equation*}
    d(X,p)-d(X,q)=\tau i,
\end{equation*}
where $i$ is one of integers $-3, -2, -1, 0, 1, 2, 3$.
\end{lemma}

\begin{proof}

\begin{figure}[t]
    \centering
    \includegraphics[width=0.45\hsize, keepaspectratio]{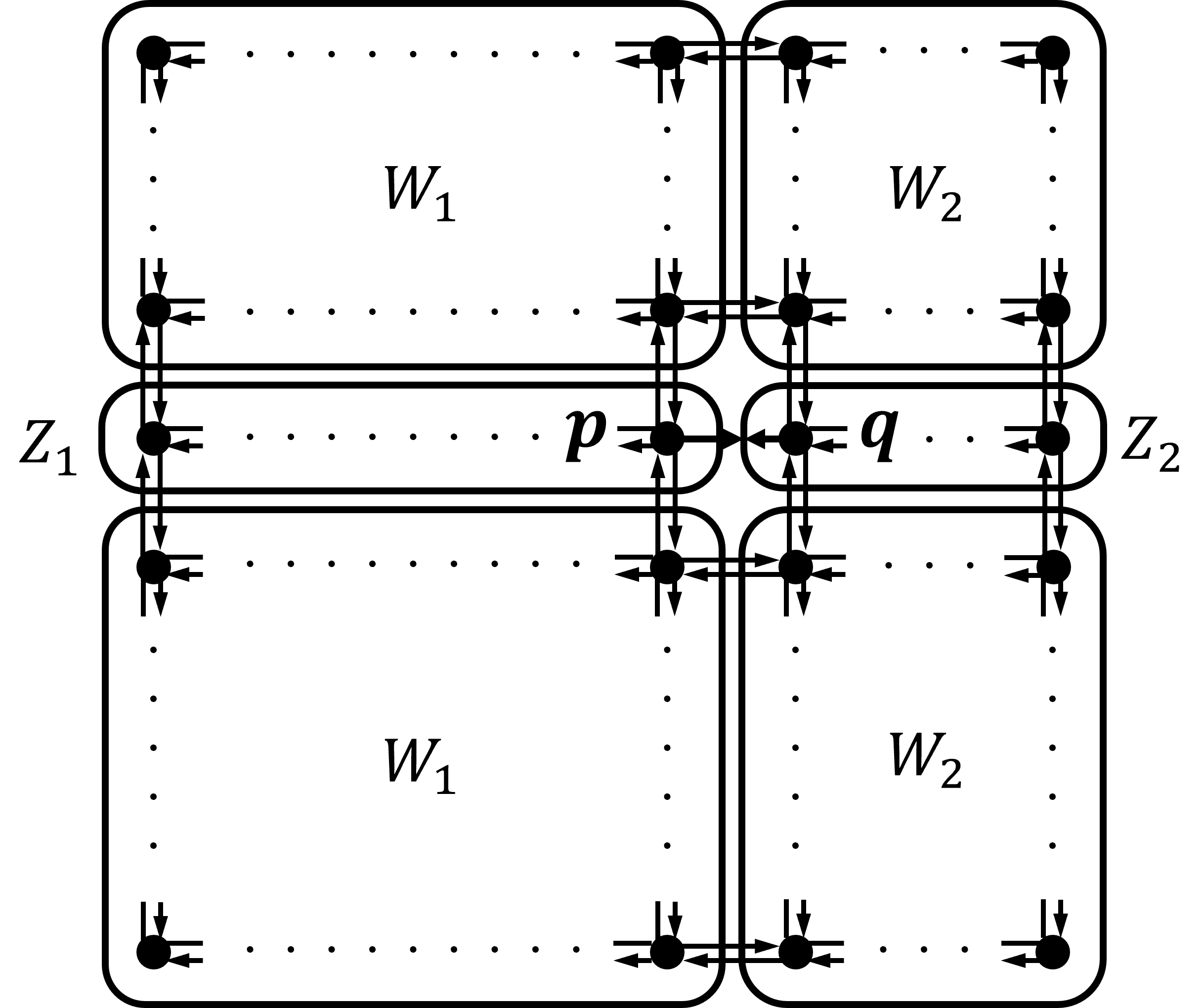}
    \caption{Division of the regions based on the inclusion of node $x$}
    \label{fig:area}
\end{figure}

Let us define subsets of node set in $\overline{\mathcal{N}}(y)$ as follows (see Fig.~\ref{fig:area}):
\begin{equation*}
\begin{split}
    W_1 &\coloneqq \{(i,j) \in V(y) \mid i \leq 0, j \neq 0\},\\
    W_2 &\coloneqq \{(i,j) \in V(y) \mid i \geq 1, j \neq 0\},\\
    Z_1 &\coloneqq \{(i,j) \in V(y) \mid i \leq 0, j = 0\},\\
    Z_2 &\coloneqq \{(i,j) \in V(y) \mid i \geq 1, j = 0\}.\\
\end{split}
\end{equation*}
Furthermore, let $d(X,\{p,q\})$ denote the minimum value among $d(X,p)$ and $d(X,q)$, that is,
\begin{equation}
d(X,\{p,q\}) = \min \{d(X,p),d(X,q)\}.
\label{eq:d(X,{u,v})}
\end{equation}
Here, we know that there exists a node $x \in X$ such that
\begin{equation}
d(X,\{p,q\}) = \min \{d(x,p),d(x,q)\}.
\label{eq:d_x}
\end{equation}
Let us consider the following four cases regarding the position of such $x$: (Case~1) $x\in{W_1}$, (Case~2) $x\in{Z_1}$, (Case~3) $x\in{W_2}$, and (Case~4) $x\in{Z_2}$.
We give the proof only for Cases~1 and 2 because Cases~3 and 4 shown in a similar way.

(Case~1) When $x\in{W_1}$, we have
\begin{equation}
    d(x,q) = d(x,p) + \tau > d(x,p).
    \label{eq:x_u=x_v+tau}
\end{equation}
Then by equation \eqref{eq:d(X,{u,v})}, \eqref{eq:d_x} and \eqref{eq:x_u=x_v+tau}, we have
\begin{equation}
    d(x,p) = d(X,\{p,q\}) \leq d(X,p).
    \label{eq:d(X,p)=d(x,p)}
\end{equation}
Since $d(X,p) \leq d(x,p)$ holds, we have
\begin{equation}
    d(x,p) = d(X,\{p,q\}) =d(X,p).
    \label{eq:d(X,p)_1}
\end{equation}
Equation \eqref{eq:d(X,p)_1} implies that
\begin{equation}
    d(X,p) \leq d(X,q)
    \label{eq:d(X,p)}
\end{equation}
holds.
Furthermore, since $d(X,q) \leq d(x,q)$, we have
\begin{equation}
    d(X,q) \leq d(x,p) + \tau = d(X,p) + \tau.
    \label{eq:W_1}
\end{equation}
Considering equations \eqref{eq:d(X,p)} and \eqref{eq:W_1}, we obtain
\begin{equation}
\begin{split}
    -\tau \leq d(X,p)-d(X,q) \leq 0.
    \label{eq:A_W_1}
\end{split}
\end{equation}

(Case~2) When $x\in{Z_1}$, because we need to go through $W_1$ and $W_2$ from $x$ to $q$, we have
\begin{equation}
    d(x,q)=d(x,p) + 3\tau > d(x,p).
    \label{eq:x_u=x_v+3tau}
\end{equation}
In a similar way to Case~1, we can conclude that equations \eqref{eq:d(X,p)=d(x,p)} and \eqref{eq:d(X,p)} hold.
Since $d(X,q) \leq d(x,q)$, by equations \eqref{eq:d(X,p)=d(x,p)} and \eqref{eq:x_u=x_v+3tau}, we have
\begin{equation}
    d(X,q) \leq d(X,p) + 3\tau.
    \label{eq:Z_1}
\end{equation}
Combining equations \eqref{eq:d(X,p)} and \eqref{eq:Z_1}, we obtain
\begin{equation}
\begin{split}
    -3\tau \leq d(X,p)-d(X,q) \leq 0.
    \label{eq:A_Z_1}
\end{split}
\end{equation}

(Case~3) When $x\in{W_2}$, similar to Case~1, we obtain
\begin{equation}
\begin{split}
    0 \leq d(X,p)-d(X,q) \leq \tau.
    \label{eq:A_W_2}
\end{split}
\end{equation}

(Case~4) When $x\in{Z_2}$, similar to Case~2, we obtain
\begin{equation}
\begin{split}
    0 \leq d(X,p)-d(X,q) \leq 3\tau.
    \label{eq:A_Z_2}
\end{split}
\end{equation}

Summarizing the above cases, for all node sets $X \subseteq S^+$, we have
\begin{equation}
    -3\tau \leq d(X,p)-d(X,q) \leq 3\tau.
    \label{eq:-3to3}
\end{equation}
Recall that $d(X,p)$ and $d(X,q)$ are integer multiples of $\tau$, which concludes the proof.
\qed
\end{proof}

Recall that $d(X,p)$ can take at most $O(\sqrt{n})$ values, and for a fixed value of $d(X,p)$, the number of values which $d(X,q)$ can take is at most seven by Lemma~\ref{combination}. Therefore, we obtain the following theorem. 

\begin{theorem}\label{thm:family}
The number of dominant source sets in $\mathcal{X}$ is $O(\sqrt{n})$.
\end{theorem}

\subsection{Algorithms}
\label{sec:algorithm}
Our algorithm consists of three steps.
First, we compute $w(X_{i,j})$ for all $X_{i,j} \in \mathcal{X}$.
Next, we calculate the segments that compose $\Theta(X_{i,j},y)$ for all $X_{i,j} \in \mathcal{X}$.
Finally, we find the value of $y$ that minimizes the upper envelope of the function family $\{\Theta(X_{i,j},y) \mid X_{i,j} \in \mathcal{X}\}$.

Let us first describe a method to compute $w(X_{i,j})$ for all $X_{i,j} \in \mathcal{X}$.
Since the number of nodes is $n$, we can compute $w(X_{i,j})$ in $O(n)$ time for any $i, j$.
Moreover, we have $|\mathcal{X}|=O(\sqrt{n})$ by Theorem~\ref{thm:family}, so all $w(X_{i,j})$ can be computed in $O(n\sqrt{n})$ time.
However, $w(X_{i,j})$ can be calculated more efficiently by using dynamic programming.
See the following lemma.
\begin{figure*}[tb]
    \centering
    \includegraphics[width=0.5\hsize, keepaspectratio]{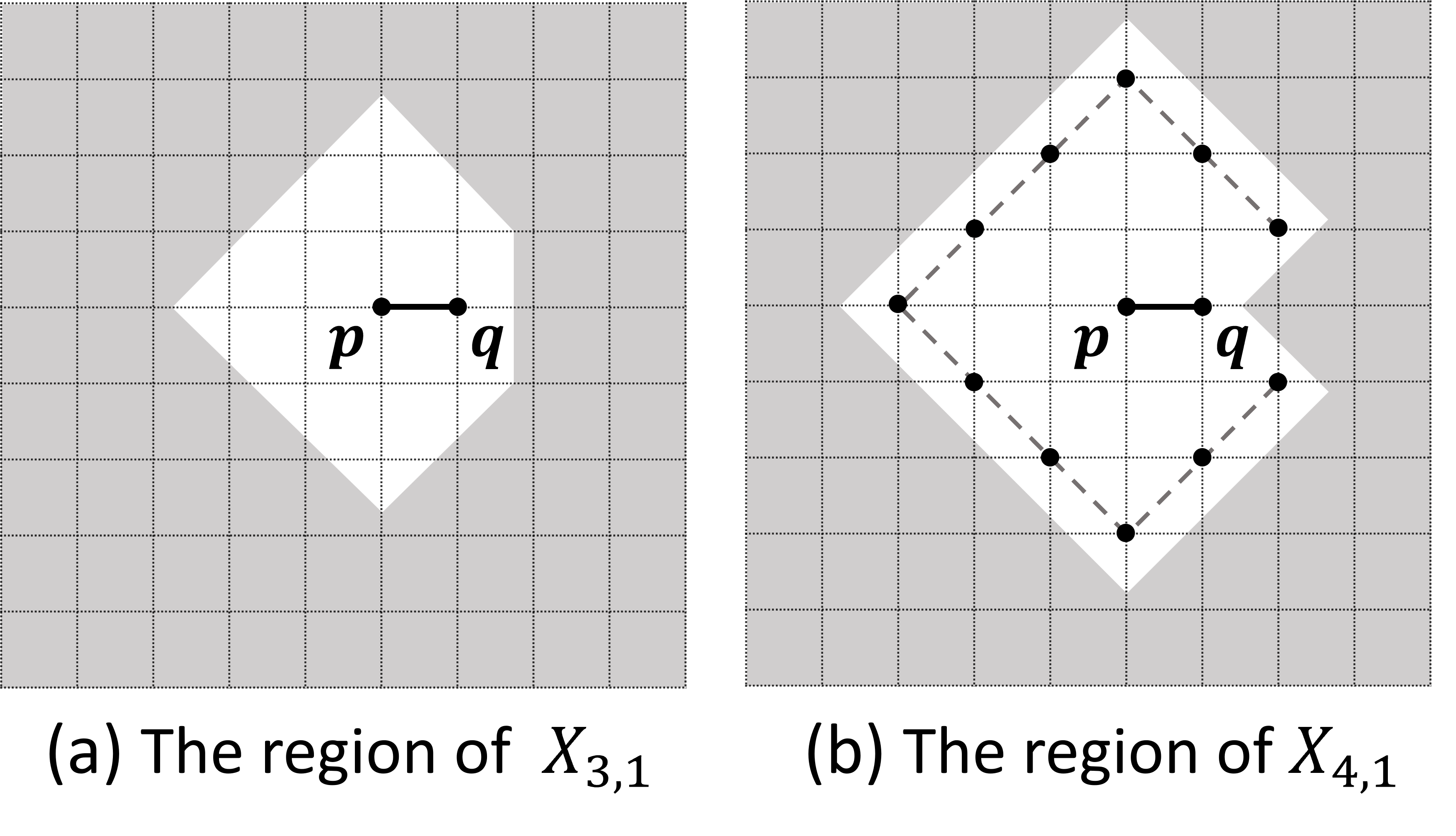}
    \caption{The regions of $X_{3,1}$ and $X_{4,1}$ consist of nodes within the gray area respectively.
    The nodes located on the dashed line in (b) represent the difference between $X_{3,1}$ and $X_{4,1}$.}
    \label{fig:w(X_i,j)}
\end{figure*}

\begin{lemma}\label{lem:w(X_{i,j})}
Values $w(X_{i,j})$ for all $X_{i,j} \in \mathcal{X}$ can be computed in $O(n)$ time.
\end{lemma}

\begin{proof}
Basically we calculate $w(X_{i,j})$ by subtracting the sum of supplies of $O(\sqrt{n})$ nodes from $w(X_{i-1,j})$ or $w(X_{i,j-1})$ as shown in Fig.~\ref{fig:w(X_i,j)}, which can be done in $O(\sqrt{n})$ time.
We then show the order of calculation of $w(X_{i,j})$ as follows.
First of all, we calculate $w(X_{0,0})$ in $O(n)$ time, and successively updating from $w(X_{0,0})$, we calculate $w(X_{1,0})$ and $w(X_{2,0})$.
For $i=0,1,2$, successively updating from $w(X_{i,0})$, we calculate $w(X_{i,1}), \ldots, w(X_{i,i+3})$.
For $i \geq 3$, we update $w(X_{i-1,i-3})$ to $w(X_{i,i-3})$, and successively updating from $w(X_{i,i-3})$, calculate $w(X_{i,i-2}), \ldots, w(X_{i,i+3})$.
Note that for each $i$, we consider only at most seven subsets $X_{i,j}$ according to Lemma~\ref{combination}.

By Theorem~\ref{thm:family}, the total number of subsets $X_{i,j}$ considered is $O(\sqrt{n})$, therefore, all values of $w(X_{i,j})$ can be computed in $O(n) + O(\sqrt{n}) \times O(\sqrt{n}) = O(n)$ time.
\qed
\end{proof}

In the second step, we calculate the segments that compose $\Theta(X_{i,j},y)$ for all $X_{i,j} \in \mathcal{X}$.
Recall that $d(X_{i,j},p) = i \tau$ and $d(X_{i,j},q) = j \tau$.
By Theorem~\ref{thm:Theta(X,y)}, $\Theta(X_{i,j},y)$ is piecewise linear function represented by
\begin{equation}
\begin{split}
    \Theta(X_{i,j},y) = \min
    \biggl\{
    & y + \frac{w(X_{i,j})}{c} + i \tau,
    -y + \tau + \frac{w(X_{i,j})}{c} + j \tau,\\
    & \frac{w(X_{i,j})}{2c} + \frac{i + j + 1}{2}\tau
    \biggr\}.
    \label{eq:theta_X_i,j}
\end{split}
\end{equation}
Using $w(X_{i,j})$ obtained at the first step, $\Theta(X_{i,j},y)$ can be calculated in $O(1)$ time.
According to equation~\eqref{eq:new_Theta^*(y)}, $\Theta^*(y)$ is the upper envelope of the function family $\{\Theta(X_{i,j},y) \mid X_{i,j} \in \mathcal{X}\}$.
By equation~\eqref{eq:theta_X_i,j}, $\Theta(X_{i,j},y)$ consists of at most three line segments, hence by Theorem~\ref{thm:family}, $\Theta^*(y)$ is determined by the upper envelope of $O(\sqrt{n})$ line segments, which can be computed in $O(\sqrt{n} \log n)$ time by Corollary~\ref{cor:upperenvelope}.
Thus, the value of $y$ that minimizes $\Theta^*(y)$ can also be found in $O(\sqrt{n}\log n)$ time in the third step. 


Summarizing the above discussions, the first step requires $O(n)$ time, the second step requires $O(\sqrt{n})$ time, and the third step requires $O(\sqrt{n} \log n)$ time.
Therefore, we obtain the following theorem.

\begin{theorem}\label{thm:main}
Given a dynamic flow grid network of $N \times N$ nodes (without any sink) with uniform edge capacity and transit time,  the 1-sink location problem on a particular edge can be solved in $O(n)$ time.
\end{theorem}

Solving problem {\sf SLE} for each edge, the 1-sink location problem for all edges can be solved in $O(n^2)$ time.
%
%
In general, the evacuation completion time can be computed in $\tilde{O}(m^2 k^4 + m^2 nk)$ time using the algorithm developed by Schlöter et al.~\cite{schloter2022faster}.
Here, $m$ is the number of edges and $k$ is the total number of sources and sinks, and the $\tilde{O}$ omits the logarithmic factor from Big~$O$ notation.
Adapting this algorithm to our problem, the evacuation completion time for each node can be computed in $\tilde{O}(n^6)$ time since $m=O(n)$ and $k=O(n)$.
If we focus on the sink location on only nodes, by applying this algorithm to each node, the optimal 1-sink location can be computed in $\tilde{O}(n^7)$ time. 
Therefore, we obtain the following corollary.

\begin{corollary}
Given a dynamic flow grid network of $N \times N$ nodes (without any sink) with uniform edge capacity and transit time, 
the 1-sink location problem can be solved in $\tilde{O}(n^7)$ time.
\end{corollary}

\section{Conclusion}\label{sec:conclusion}
The sink location problem is a kind of evacuation problems that seeks to determine the location of sinks that minimize the time it takes for each evacuee to arrive at one of the evacuation facilities.
In this paper, we proposed a polynomial-time algorithm for the 1-sink location problem on a dynamic flow grid network with uniform edge capacity and transit time.
This is the first polynomial-time algorithm for networks that contain a number of cycles.
We remark that our provided approach can be extended to the 1-sink location problem on $M \times N$ grid networks, which gives a polynomial-time algorithm even when $M \neq N$.

It would be interesting to develop polynomial-time algorithms for the sink location problems on more complex networks, especially, a grid networks with multiple number of capacities and transit times.
In addition, from the viewpoint of real world applications, developing polynomial-time algorithms for grid networks with holes or the case of multiple sinks is a future challenge.



%
%
%
\bibliographystyle{splncs04}
\bibliography{sample}
%


\end{document}